\documentclass[aps,twocolumn,pra,10pt,floatfix,showpacs]{revtex4-1}

\usepackage[sc]{mathpazo}
\usepackage{amsmath}
\usepackage{amssymb}
\usepackage{amsthm}
\usepackage{lmodern}
\usepackage{enumitem}
\usepackage{color}

\usepackage{stmaryrd}

\RequirePackage{filecontents}        

\usepackage{graphicx}
\usepackage{float}
\usepackage{tikz}
\usetikzlibrary{arrows.meta,shapes.geometric,shapes.misc}

\makeatletter
\newsavebox\myboxA
\newsavebox\myboxB
\newlength\mylenA
\newcommand*\xoverline[2][0.75]{%
    \sbox{\myboxA}{$\m@th#2$}%
    \setbox\myboxB\null
    \ht\myboxB=\ht\myboxA%
    \dp\myboxB=\dp\myboxA%
    \wd\myboxB=#1\wd\myboxA
    \sbox\myboxB{$\m@th\overline{\copy\myboxB}$}
    \setlength\mylenA{\the\wd\myboxA}
    \addtolength\mylenA{-\the\wd\myboxB}%
    \ifdim\wd\myboxB<\wd\myboxA%
       \rlap{\hskip 0.5\mylenA\usebox\myboxB}{\usebox\myboxA}%
    \else
        \hskip -0.5\mylenA\rlap{\usebox\myboxA}{\hskip 0.5\mylenA\usebox\myboxB}%
    \fi}
\makeatother

\let\emptyset\varnothing

\def\to{\rightarrow}

\def \lket {\left|}
\def \rket {\right\rangle}
\def \lbra {\left\langle}
\def \rbra {\right|}
\newcommand{\ket}[1]{\lket\mspace{0.5mu} #1 \mspace{0.5mu}\rket}

\newcommand{\bra}[1]{\lbra\mspace{0.5mu} #1 \mspace{0.5mu}\rbra}

\def\dsl{\llbracket}
\def\dsr{\rrbracket}

\def\cz{\sansserif{CZ}}

\def\C{\mathcal{C}}

\def\H{\mathcal{H}}

\def\N{\mathcal{N}}

\def\S{\mathcal{S}}

\def\pa{\sansserif{A}}
\def\pb{\sansserif{B}}
\def\pc{\sansserif{C}}
\def\i{\sansserif{I}}
\def\pp{\sansserif{P}}
\def\pq{\sansserif{Q}}

\def\px{\sansserif{X}}
\def\py{\sansserif{Y}}
\def\pz{\sansserif{Z}}
\def\pq{\sansserif{Q}}

\def\cnot{\sansserif{CNOT}}

\newtheorem*{theorem*}{Theorem}

\usepackage{thmtools}
\usepackage{thm-restate}

\newcommand{\sansserif}[1]{%
  \ifmmode
  \mathsf{#1}%
  \else
   \textsf{#1}%
  \fi
}
\usepackage{relsize}
\begin{document}
\title{Topological wormholes}
\author{Anirudh Krishna}
\affiliation{
	D\'epartment de physique \& Institut Quantique,
	Universit\'e de Sherbrooke,
	Sherbrooke, Qu\'ebec, Canada J1K 2R1}
\author{David Poulin}
\affiliation{
	D\'epartment de physique \& Institut Quantique,
	Universit\'e de Sherbrooke,
	Sherbrooke, Qu\'ebec, Canada J1K 2R1}
\begin{abstract}
	Locality plays a fundamental role in quantum computation but also severely restricts our ability to store and process quantum information.
	We argue that this restriction may be unwarranted and re-examine quantum error correcting codes.
	We proceed to introduce new defects on the surface code called wormholes.
	These novel defects entangle two spatially separated sectors of the lattice.
	When anyonic excitations enter the mouth of a wormhole, they emerge through the other mouth.
	Wormholes thus serve to connect two spatially separated sectors of a flat, $2$D lattice.
	We show that these defects are capable of encoding logical qubits and can be used to perform all gates in the Clifford group.
\end{abstract}
\date{\small \today}
\maketitle

\emph{Introduction.---}
Locality plays a central role in condensed matter and quantum information science.
This is exemplified by Kitaev's toric code \cite{kitaev2003fault}, its variants \cite{bravyi1998quantum} and the color code \cite{bombin2006topological}.
These physical models are defined by Hamiltonians composed entirely of local terms of low weight and yet display topological order -- the ground states of the system cannot be discerned using local measurements.
From the perspective of quantum computation, this means that the code space is robust to local errors as perturbations must collude over a large distance to induce a logical error.
Furthermore, the locality of the Hamiltonian is a boon for experimental realizations of quantum error correction and considerably simplifies the syndrome extraction circuits.
For this and many other reasons, such models are a promising blueprint for scalable quantum computers \cite{FowlerEtAl12}.

Locality however poses severe restrictions on storing and processing encoded quantum information \cite{bravyi2010tradeoffs,bravyi2013classification,webster2018braiding}.
Architectures based on Nitrogen Vacancy (NV) centers and ion traps have softer constraints on coupling qubits that are not adjacent \cite{nickerson2013topological}.
With the advent of deterministic methods to share entanglement non-locally on superconducting qubit architectures \cite{kurpiers2018deterministic,axline2018demand,campagne2018deterministic}, the strict restriction of locality in the design of quantum error correcting codes may be unwarranted.
Furthermore, modeling qubits as point-like objects may not apply to physical implementations which use extended objects such as a resonator to store quantum information \cite{katzgraber2014glassy}.
In such an architecture, the geometry of qubit couplings may not be suitably represented by a two-dimensional grid.
For these reasons, we have chosen to bend the rules of locality and re-examine quantum error correcting codes.
We emphasize that we are not addressing the no-go result of Bravyi et al. \cite{bravyi2010tradeoffs} which places restrictions on two-dimensional quantum memories; nor are we violating the Bravyi-Koenig bound \cite{bravyi2013classification} which limits what gates can be performed fault tolerantly on two-dimensional codes.

In this article, we demonstrate how introducing a small amount of non-locality gives rise to novel defects on the toric code which we call \emph{wormholes}.
These defects possess two mouths that connect two spatially separated sectors of the lattice that we refer to as the mouths of the wormhole.
The name is motivated by considering the movement of lattice excitations called anyons.
If an anyon were to enter the mouth of a wormhole, it emerges via the other.
In turn, this means that two anyons can be spatially separated by an arbitrary distance and still share entanglement via the wormhole.
Furthermore, we shall show that we can use wormholes to encode logical information.
We then demonstrate how we can perform all Clifford operations on the logical information.
This can be seen as a unification of previous defect based encoding schemes combining puncture and twist defects \cite{raussendorf2007fault,Bombin10,Bombin11,BrownEtAl17}.

If we eschew locality entirely, we can obtain quantum Low Density Parity Check (LDPC) codes that are capable of encoding a number of qubits that grows with the block size \cite{FreedmanMeyerLuo02,GuthLubotzky14,TillichZemor09}.
However, engineering such connections may be infeasible with current technology.
Therefore, this work can additionally be seen as a proposal for codes in the spectrum between entirely local codes on the one end and quantum LDPC codes on the other.

The techniques proposed here to perform gates generalize to a powerful class of quantum LDPC codes called hypergraph product codes.
In a companion paper \cite{krishna2019fault}, we outline how to generalize the techniques presented here to perform gates on a certain class of LDPC codes called hypergraph product codes \cite{TillichZemor09,LeverrierTillichZemor15}.
This is the first technique to fault tolerantly perform gates on this family of error correcting codes.

Lastly, our construction could contribute to the discussion on the connection between entanglement and the geometry of spacetime \cite{pastawski2015holographic}.
We first note that the geometry that an anyon experiences is dictated by the entanglement in the underlying spin substrate.
This is reminiscent of the ER = EPR conjecture in quantum gravity \cite{susskind2016copenhagen}.
Secondly, if we define the entropy of a wormhole as the entanglement entropy between two mouths, then we find that it scales with the size of the boundary of the mouths rather than the size of the mouths.
This mirrors the Bekenstein entropy \cite{bekenstein1973black} which also scales proportionally to the area of a black hole.

\emph{Background and notation. ---}
The toric code \cite{bravyi1998quantum} is a quantum error correcting code defined on a square lattice with periodic boundary conditions.
The qubits are placed on the edges of the lattice and the vertices and plaquettes of the lattice serve to define a local Hamiltonian $\H$.
We introduce some notation at this juncture which will make it convenient to represent these objects.
\begin{center}
	\begin{tikzpicture}[scale=0.6]
		\draw (-2.5,-0.25) edge[-] (-2.5,1.25);
		\draw (-3.25,0.5) edge[-] (-1.75,0.5);
		\draw (-2.7,0.3) rectangle (-2.3,0.7);
		\node at (-1.25,0.5) {$:=$};
		\draw (0,-0.25) edge[-] (0,1.25);
		\draw (-0.75,0.5) edge[-] (0.75,0.5);
		\draw (0,0.125) circle (0.15cm);
		\draw (0,0.875) circle (0.15cm);
		\draw (-0.375,0.5) circle (0.15cm);
		\draw (0.375,0.5) circle (0.15cm);
		\draw (1.875,0) rectangle (2.875,1);
		\fill[black] (2.175,0.3) rectangle (2.575,0.7);
		\node at (3.5,0.5) {$:=$};
		\draw (4.25,0) rectangle (5.25,1);
		\fill[black] (4.75,0) circle (0.2cm);
		\fill[black] (4.75,1) circle (0.2cm);
		\fill[black] (4.25,0.5) circle (0.2cm);
		\fill[black] (5.25,0.5) circle (0.2cm);
	\end{tikzpicture}
\end{center}
\vspace{-2mm}
Circles on edges represent single-qubit operators on the corresponding qubits -- empty circles represent single-qubit Pauil $\px$ operators and filled circles represent single-qubit Pauli $\pz$ operators.
We will later see multiple circles arranged in some pattern and this corresponds to products of the respective Pauli operators.
Likewise, each local term or \emph{stabilizer} is a product of Pauli operators and is denoted using a square node -- empty square nodes on the vertices represent $\px$ stabilizers, and dark square nodes on plaquettes represent $\pz$ stabilizers. The Hamiltonian is then expressed as the sum of these local terms,
\begin{center}
	\begin{tikzpicture}[scale=0.6]
		\node at (-2.5,0.3) {$\H = -\mathlarger{\sum}_{+}$};
		\draw (-0.5,-0.25) edge[-] (-0.5,1.25);
		\draw (-1.25,0.5) edge[-] (0.25,0.5);
		\draw (-0.7,0.3) rectangle (-0.3,0.7);
		\node at (1,0.3) {$- \mathlarger{\sum}_{\square}$};
		\draw (1.875,0) rectangle (2.875,1);
		\fill[black] (2.175,0.3) rectangle (2.575,0.7);
		\node at (3.5,0.5) {$.$};
	\end{tikzpicture}
\end{center}
\vspace{-2mm}
The first sum is over all the vertices of the lattice and the second sum is over all the faces of the lattice.
The code space is the ground space of $\H$.

\emph{Surface code defects. ---}
The Clifford group is the set of unitary gates which is generated by the Hadamard, phase and $\cnot$ gates.
These gates occupy a special role in the theory of fault tolerance and quantum error correction.
We begin by describing defects on the toric code that are capable of encoding qubits in a manner that facilitates Clifford gates.

Punctures are defects on the surface code that come in one of two varieties, smooth and rough, as shown in fig.\ \ref{fig:basicSmoothRough}.
A smooth puncture is created by measuring $\px$ on the support of a set of $\px$ stabilizer generators whereas a rough puncture is created by measuring $\pz$ on the support of a set of $\pz$ stabilizer generators.
A pair of smooth or rough punctures can be used to encode a logical qubit as shown in fig.\ \ref{fig:encoding} (a).
The logical $\pz$ ($\px$) assigned to a pair of smooth (rough) punctures is a loop of $\pz$'s ($\px$'s) encircling a puncture;
the conjugate logical $\px$ ($\pz$) operator is a chain of $\px$'s ($\pz$'s) between two smooth (rough) punctures.

\emph{Braiding} punctures results in a logical $\cnot$ with the smooth puncture serving as control and the rough puncture serving as target \cite{raussendorf2007fault}.
However, braiding is limited; as such, it maps $\px$ operators to $\px$ operators and $\pz$ operators to $\pz$ operators.
We need to break this restriction to perform a broader class of gates.
\begin{figure}[htp]
	\centering
	\begin{tikzpicture}[rotate=90,scale=.35]
      \draw[step=1cm,gray] (0,-2) grid (7,8);
      \draw (1.5,4) circle (0.2cm);
      \draw (1.5,5) circle (0.2cm);
      \draw (1.5,6) circle (0.2cm);
      \draw (2,3.5) circle (0.2cm);
      \draw (2,4.5) circle (0.2cm);
      \draw (2,5.5) circle (0.2cm);
      \draw (2,6.5) circle (0.2cm);
      \draw (2.5,4) circle (0.2cm);
      \draw (2.5,5) circle (0.2cm);
      \draw (2.5,6) circle (0.2cm);
      \draw (3,3.5) circle (0.2cm);
      \draw (3,4.5) circle (0.2cm);
      \draw (3,5.5) circle (0.2cm);
      \draw (3,6.5) circle (0.2cm);
      \draw (3.5,4) circle (0.2cm);
      \draw (3.5,5) circle (0.2cm);
      \draw (3.5,6) circle (0.2cm);
      \draw (4,3.5) circle (0.2cm);
      \draw (4,4.5) circle (0.2cm);
      \draw (4,5.5) circle (0.2cm);
      \draw (4,6.5) circle (0.2cm);
      \draw (4.5,4) circle (0.2cm);
      \draw (4.5,5) circle (0.2cm);
      \draw (4.5,6) circle (0.2cm);
      \fill[black] (3,1.5) circle (0.2cm);
      \fill[black] (3,0.5) circle (0.2cm);
      \fill[black] (3,-0.5) circle (0.2cm);
      \fill[black] (3.5,2) circle (0.2cm);
      \fill[black] (3.5,1) circle (0.2cm);
      \fill[black] (3.5,0) circle (0.2cm);
      \fill[black] (3.5,-1) circle (0.2cm);
      \fill[black] (4,1.5) circle (0.2cm);
      \fill[black] (4,0.5) circle (0.2cm);
      \fill[black] (4,-0.5) circle (0.2cm);
      \fill[black] (4.5,2) circle (0.2cm);
      \fill[black] (4.5,1) circle (0.2cm);
      \fill[black] (4.5,0) circle (0.2cm);
      \fill[black] (4.5,-1) circle (0.2cm);
      \fill[black] (5,1.5) circle (0.2cm);
      \fill[black] (5,0.5) circle (0.2cm);
      \fill[black] (5,-0.5) circle (0.2cm);
      \fill[black] (5.5,2) circle (0.2cm);
      \fill[black] (5.5,1) circle (0.2cm);
      \fill[black] (5.5,0) circle (0.2cm);
      \fill[black] (5.5,-1) circle (0.2cm);
      \fill[black] (6,1.5) circle (0.2cm);
      \fill[black] (6,0.5) circle (0.2cm);
      \fill[black] (6,-0.5) circle (0.2cm);
      \node at (3.5,-8) {\includegraphics[scale=0.16]{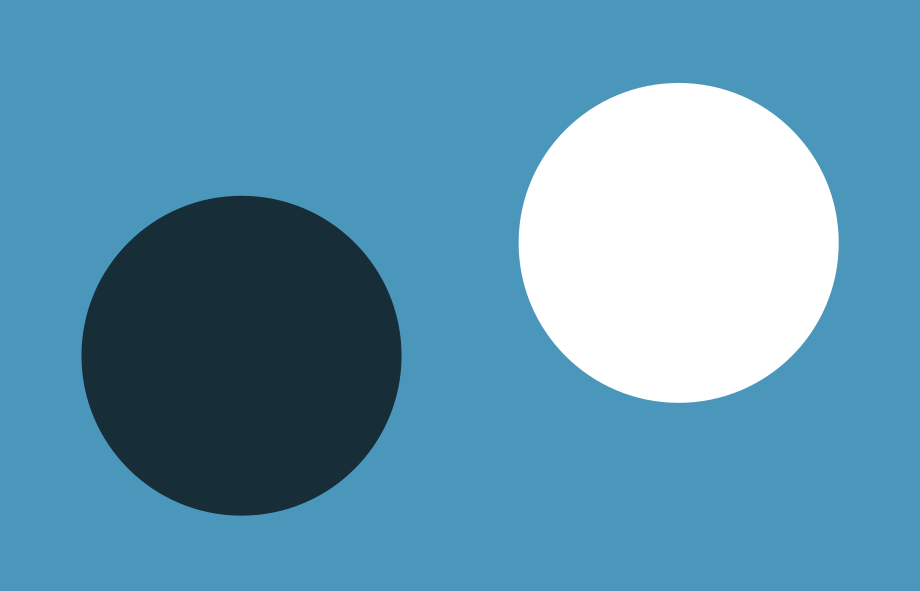}};
  	  \node at (6.1,9) {$(a)$};
	  \node at (6.1,-3.2) {$(b)$};
	\end{tikzpicture}
	\caption{Smooth and rough punctures on a lattice.
	(a) Measurements that serve to create the smooth and rough punctures are denoted using empty and filled circles resp.
	(b) The corresponding lattice-free representation with smooth puncture in black and rough puncture in white.}
	\label{fig:basicSmoothRough}
\end{figure}
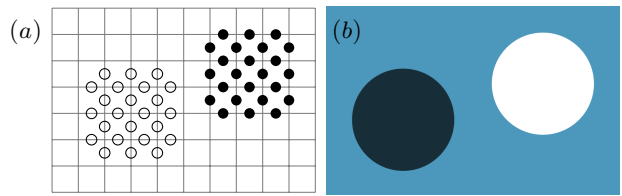

Twists are yet another defect on the surface code that address this issue \cite{Bombin10}.
These objects are created by measuring two-qubit operators comprised of one $\px$ and one $\pz$ on adjacent qubits.
The measurement is depicted in fig. \ref{fig:pairMergeAdjacent} (a) by the two circle nodes and the line connecting them.
The individual plaquette and vertex stabilizer generators incident to these qubits anti-commute with this measurement.
This pair of stabilizers is repaced by its product to resolve this frustration.
It is depicted by the line connecting the two square nodes in fig. \ref{fig:pairMergeAdjacent} (a).
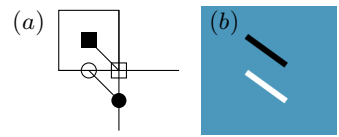
\begin{figure}[htp]
	\centering
	\begin{tikzpicture}[scale=0.4]
    \draw (0,0) rectangle (2,2);
    \fill[black] (0.75,0.75) rectangle (1.25,1.25);
    \draw (0,0) edge[-] (4,0);
    \draw (2,2) edge[-] (2,-2);
    \draw (1.75,-0.25) rectangle (2.25,0.25);
    \draw (1,1) edge[-] (2,0);

    \draw (1,0) circle (0.25cm);
    \fill[black] (2,-1) circle (0.25cm);
    \draw (1,0) edge[-] (2,-1);

    \node at (7,0) {\includegraphics[scale=0.3]{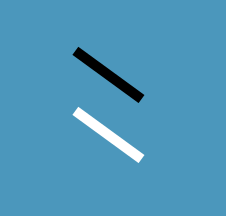}};
    \node at (-1,1.6) {$(a)$};
    \node at (5.5,1.6) {$(b)$};
\end{tikzpicture}
\caption{
(a) Two-qubit measurement indicated by two circle nodes and line connecting them.
Hybrid stabilizer indicated by square nodes and line connecting them.
(b) Lattice-free representation of the two-qubit measurement, indicated by a white line, and the product of stabilizers, indicated by a black line.}
\label{fig:pairMergeAdjacent}
\end{figure}
As shown in fig.\ \ref{fig:twist}, we can perform these two-qubit operators along a line referred to as a defect line; the hybrid stabilizers at either end of this line are called twists.
$\px$ and $\pz$ stabilizer generators across the defect line pair up to form hybrid stabilizers of weight $6$.
These objects can be used to supplement the set of possible operations on punctures -- a smooth puncture that crosses the defect line is transformed into a rough puncture and vice versa.
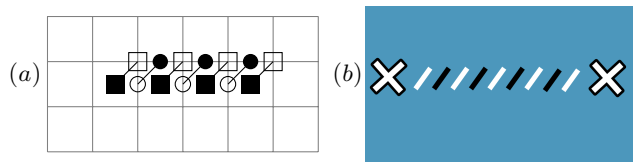
\begin{figure}[htp]
	\centering
	\begin{tikzpicture}[rotate=90,scale=0.6]
	    \draw[step=1cm,gray] (1,0) grid (4,6);
	    \draw (2.5,4.5) edge[-] (3,4);
	    \fill[black] (2.3,4.3) rectangle (2.7,4.7);
	    \draw (2.8,3.8) rectangle (3.2,4.2);
	    \draw (2.5,3.5) edge[-] (3,3);
	    \fill[black] (2.3,3.3) rectangle (2.7,3.7);
	    \draw (2.8,2.8) rectangle (3.2,3.2);
	    \draw (2.5,2.5) edge[-] (3,2);
	    \fill[black] (2.3,2.3) rectangle (2.7,2.7);
	    \draw (2.8,1.8) rectangle (3.2,2.2);
	    \draw (2.5,1.5) edge[-] (3,1);
	    \fill[black] (2.3,1.3) rectangle (2.7,1.7);
	    \draw (2.8,0.8) rectangle (3.2,1.2);
	    \draw (2.5,4) edge[-] (3,3.5);
	    \draw (2.5,4) circle (0.17cm);
	    \fill[black] (3,3.5) circle (0.17cm);
	    \draw (2.5,3) edge[-] (3,2.5);
	    \draw (2.5,3) circle (0.17cm);
	    \fill[black] (3,2.5) circle (0.17cm);
	    \draw (2.5,2) edge[-] (3,1.5);
	    \draw (2.5,2) circle (0.17cm);
	    \fill[black] (3,1.5) circle (0.17cm);
	    \node at (2.5,-4) {\includegraphics[angle=90,scale=0.13]{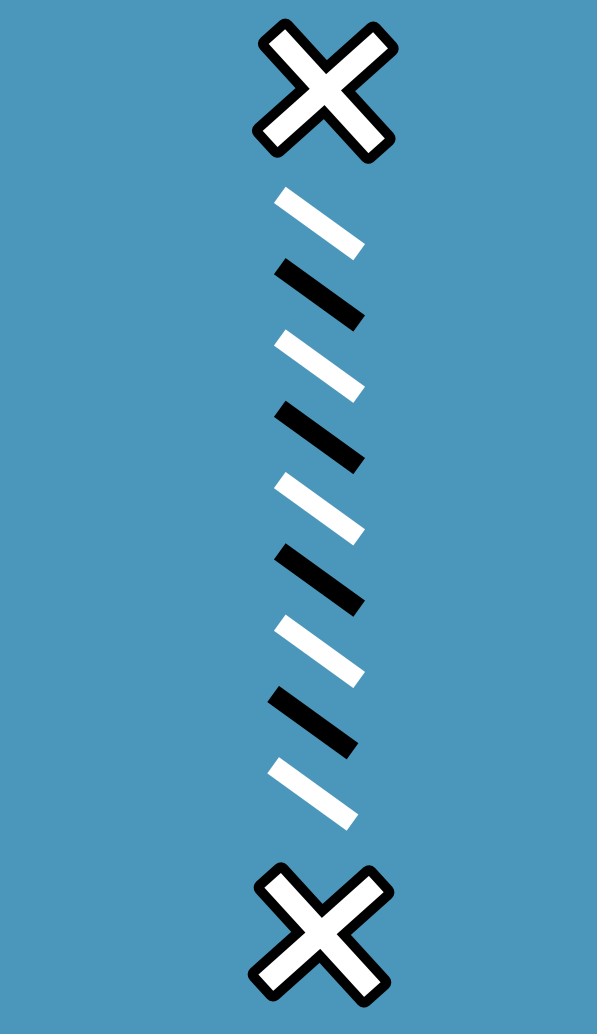}};
	    \node at (2.7,6.5) {$(a)$};
	    \node at (2.7,-0.65) {$(b)$};
	\end{tikzpicture}
	\caption{A twist on a lattice.
	(a) Measure the pairs of qubits using the two-qubit operator $\px \otimes \pz$;
	$\px$ on the horizontal edges and $\pz$ on the vertical edges.
	(b) Lattice-free representation of the twist.
	The twists are marked as white crosses.}
	\label{fig:twist}
\end{figure}
Furthermore, two pairs of twists can be used to encode a logical qubit in their own right as shown in fig.\ \ref{fig:encoding}(b).
The logical $\pz$ is the loop of $\pz$'s encircling a pair of twists and the shared defect line.
The logical $\px$ is a loop that runs between the pair that contains both $\px$ and $\pz$ operators.
We can perform single-qubit Clifford gates on encoded qubits by exchanging twists \cite{Bombin10,Bombin11,yoder2017surface,BrownEtAl17,zheng2015demonstrating}.

\begin{figure}[htp]
	\centering
	\begin{tikzpicture}
		\node at (0,0) {\includegraphics[scale=0.215]{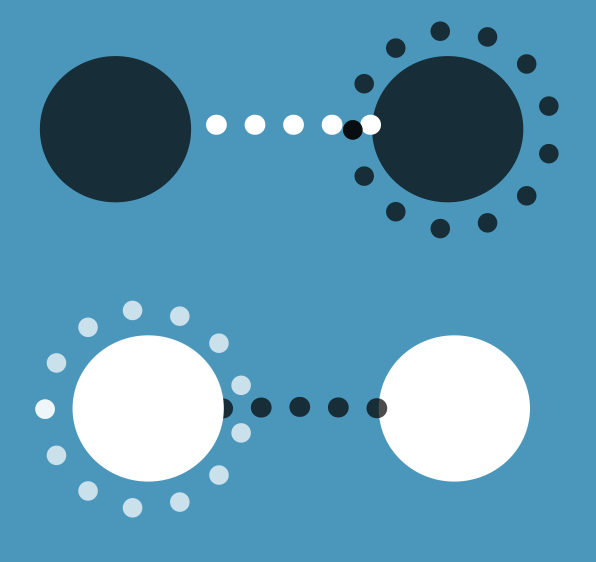}};
		\node at (4.1,0) {\includegraphics[scale=0.12]{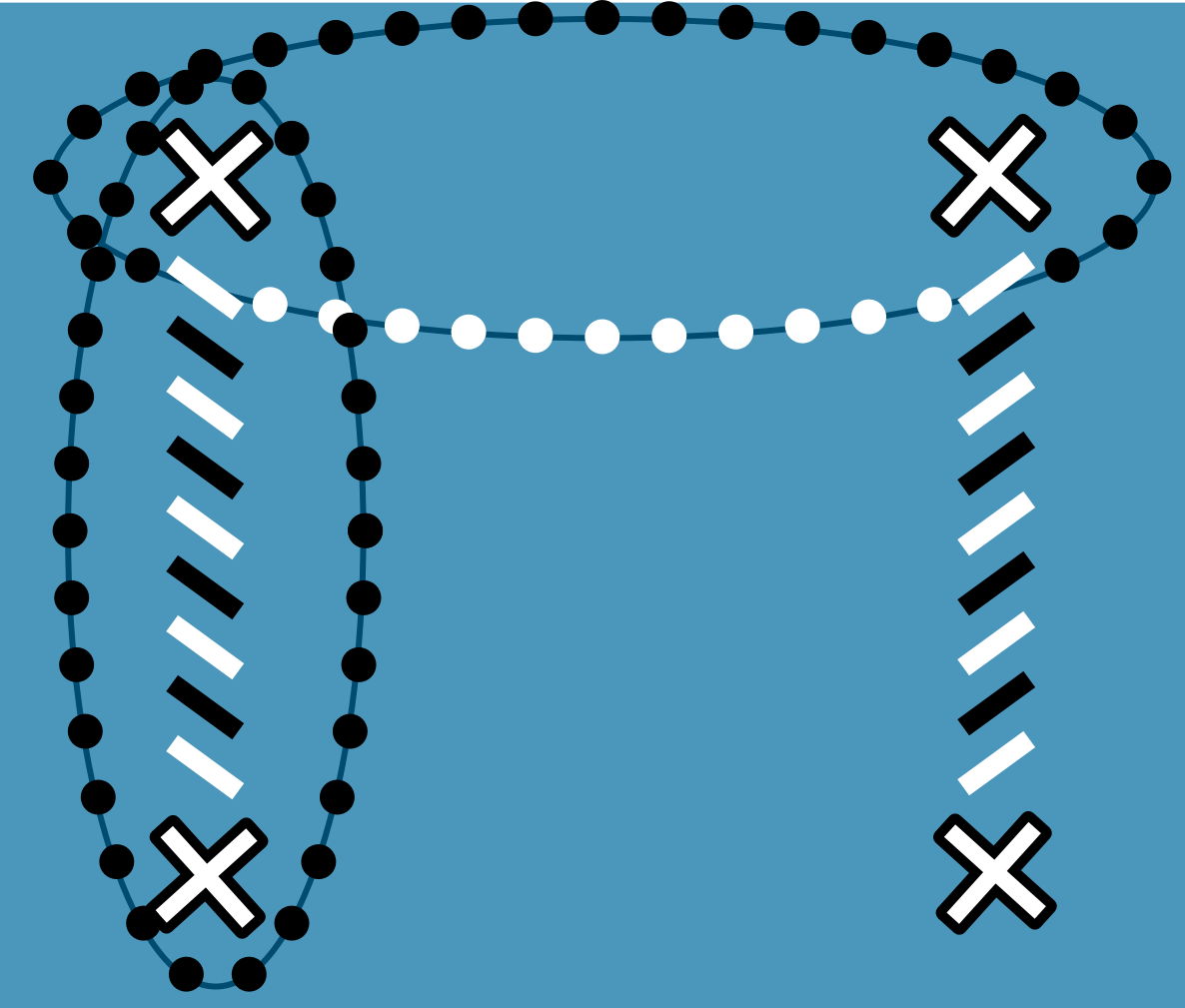}};
		\node at (-1.9,1.3) {$(a)$};
		\node at (2,1.3) {$(b)$};
	\end{tikzpicture}
	\caption{
		(a) Encoding a logical qubit in a pair of smooth or rough punctures.
		(b) Encoding a logical qubit in two pairs of twists.
	}
	\label{fig:encoding}
\end{figure}

\emph{Wormhole. ---} We introduce a new type of defect called a \emph{wormhole} that can be seen as a marriage of puncture and twist defects.
Consider the two-qubit measurement that was used to create a twist but with spatially separated $\px$ (white circle) and $\pz$ (dark circle) operators as shown.
\begin{center}
\begin{tikzpicture}[scale=0.4]
	\draw (0,0) rectangle (2,2);
    \fill[black] (0.75,0.75) rectangle (1.25,1.25);
    \draw (1,0) circle (0.25cm);

    \draw (6,0.5) edge[-] (6,3.5);
    \draw (4.5,2) edge[-] (7.5,2);
    \draw (5.75,1.75) rectangle (6.25,2.25);
    \fill[black] (6,1) circle (0.25cm);
    \draw (1,0) edge[-] (6,1) (1,1) edge[-] (6,2);
\end{tikzpicture}
\end{center}
\vspace{-2mm}
These are measured on the support of a plaquette and vertex stabilizer respectively.
To relove the anti-commutation, we replace these objects by the spatially-separated hybrid stabilizer indicated by the line joining the plaquette and vertex stabilizer generators.

All the hybrid stabilizers are a product of one plaquette and one vertex generator and thus this code remains LDPC.
We can go further by noting that there is no reason to restrict ourselves to measurements along a line.
We can measure two-qubit operators along the boundaries of punctures as shown in fig.\ \ref{fig:deformation}.
This creates two entangled punctures that are spatially separated that we refer to as the mouths of the wormhole.
These new hybrid stabilizers have weight $6$; this is the product of two stabilizers on the boundary, minus their support inside the puncture.
This weight can however be reduced by spreading the weight among some of the local checks adjacent to these stabilizers.

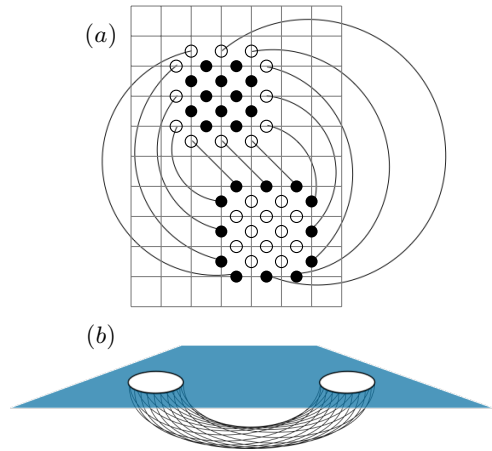
\begin{figure}[htp]
	\centering
	\begin{tikzpicture}[scale=0.4]
      \node at (4.8,3.35) {\includegraphics[scale=.32]{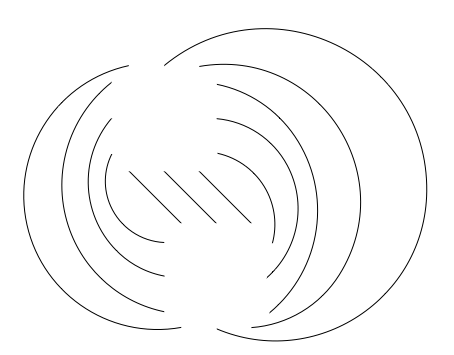}};
      \draw[step=1cm,gray] (0,-2) grid (7,8);
      \draw (1.5,4) circle (0.2cm);
      \draw (1.5,5) circle (0.2cm);
      \draw (1.5,6) circle (0.2cm);
      \draw (2,3.5) circle (0.2cm);
      \fill[black] (2,4.5) circle (0.2cm);
      \fill[black] (2,5.5) circle (0.2cm);
      \draw (2,6.5) circle (0.2cm);
      \fill[black] (2.5,4) circle (0.2cm);
      \fill[black] (2.5,5) circle (0.2cm);
      \fill[black] (2.5,6) circle (0.2cm);
      \draw (3,3.5) circle (0.2cm);
      \fill[black] (3,4.5) circle (0.2cm);
      \fill[black] (3,5.5) circle (0.2cm);
      \draw (3,6.5) circle (0.2cm);
      \fill[black] (3.5,4) circle (0.2cm);
      \fill[black] (3.5,5) circle (0.2cm);
      \fill[black] (3.5,6) circle (0.2cm);
      \draw (4,3.5) circle (0.2cm);
      \fill[black] (4,4.5) circle (0.2cm);
      \fill[black] (4,5.5) circle (0.2cm);
      \draw (4,6.5) circle (0.2cm);
      \draw (4.5,4) circle (0.2cm);
      \draw (4.5,5) circle (0.2cm);
      \draw (4.5,6) circle (0.2cm);
      \fill[black] (3,1.5) circle (0.2cm);
      \fill[black] (3,0.5) circle (0.2cm);
      \fill[black] (3,-0.5) circle (0.2cm);
      \fill[black] (3.5,2) circle (0.2cm);
      \draw (3.5,1) circle (0.2cm);
      \draw (3.5,0) circle (0.2cm);
      \fill[black] (3.5,-1) circle (0.2cm);
      \draw (4,1.5) circle (0.2cm);
      \draw (4,0.5) circle (0.2cm);
      \draw (4,-0.5) circle (0.2cm);
      \fill[black] (4.5,2) circle (0.2cm);
      \draw (4.5,1) circle (0.2cm);
      \draw (4.5,0) circle (0.2cm);
      \fill[black] (4.5,-1) circle (0.2cm);
      \draw (5,1.5) circle (0.2cm);
      \draw (5,0.5) circle (0.2cm);
      \draw (5,-0.5) circle (0.2cm);
      \fill[black] (5.5,2) circle (0.2cm);
      \draw (5.5,1) circle (0.2cm);
      \draw (5.5,0) circle (0.2cm);
      \fill[black] (5.5,-1) circle (0.2cm);
      \fill[black] (6,1.5) circle (0.2cm);
      \fill[black] (6,0.5) circle (0.2cm);
      \fill[black] (6,-0.5) circle (0.2cm);
	\node at (4,-5) {\includegraphics[scale=.32]{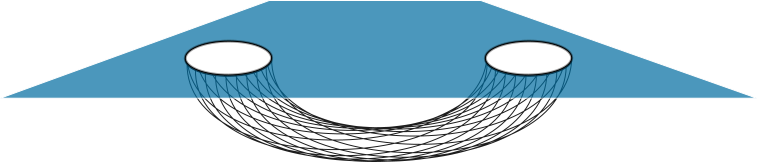}};
	\node at (-1,7) {$(a)$};
	\node at (-1,-3) {$(b)$};
	\end{tikzpicture}
	\caption{Creating a wormhole.
	(a) Measuring two-qubit Pauli operators along the boundary of a puncture.
	b) Side-view of lattice free representation of the wormhole.
		The two white circles represent the mouths of the wormhole.
		The wire mesh underneath the lattice represents the entanglement between these two patches.
		The mesh is merely a visual aid and does not represent an extension of the lattice.}
	\label{fig:deformation}
\end{figure}

The mouths of a wormhole are topologically indistinguishable.
For this reason, we drop the color of the mouths and without loss of generality, depict both mouths in white.
Upon entering one mouth, an anyon emerges via the other mouth with the opposite charge label.
Other types of wormholes are possible that preserve the topological charge of the excitations.
In general, wormhole types correspond to topological domain walls.
In \cite{barkeshli2016modular}, Barkeshli and Freedman enumerate the different boundaries that can be used to transform one type of charge to another.
In contrast, the focus of our work is to understand how to actually construct wormholes, and how to use them for performing gates on LDPC codes.

When a wormhole is created from the vacuum, it is stabilized by a pair of non-local operators shown in fig.\ \ref{fig:stab-logicals} $(a)$ \& $(b)$.
At first glance, it appears that the weight of these stabilizers scales with the size of the puncture.
However, these operators are merely products of the hybrid stabilizers on the boundary.

We can use a wormhole to encode two logical qubits as shown in fig.\ \ref{fig:stab-logicals} $(c)$ \& $(d)$, that we label $1$ and $2$.
We represent the logical $\pz$ operators as a loop of physical $\pz$ operators that encircle one mouth.
The conjugate logical operators are pairs of strings, one of $\px$ type and another of $\pz$ type that run to the mouths.
We assume that the strings terminate at a `sink' wormhole elsewhere on the lattice.


\begin{figure}[htp]
	\centering
	\begin{tikzpicture}
		\node at (0,0) {\includegraphics[width=\columnwidth]{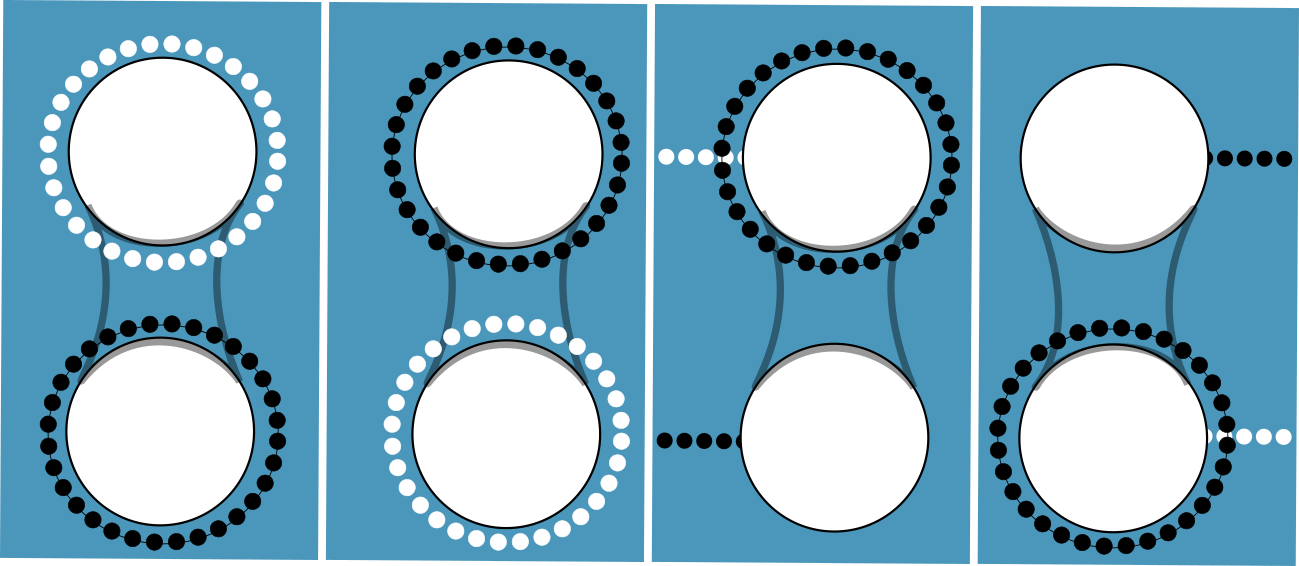}};
		\node at (-4.05,1.6) {$(a)$};
		\node at (-1.85,1.6) {$(b)$};
		\node at (0.3,1.6) {$(c)$};
		\node at (2.5,1.6) {$(d)$};
	\end{tikzpicture}
	\caption{Stabilizer and logical generators of the wormhole.
		$(a), (b)$ represent the stabilizer generators of the wormhole.
		$(c), (d)$ represent the logical operators of each logical qubit.
		The logical $\pz$ is a loop of $\pz$ operators encircling the mouth of a wormhole.
		The logical $\px$ is comprised of two strings that runs between the two mouths of the wormholes, one string of $\pz$ operators and another string of $\px$ operators.
		}
	\label{fig:stab-logicals}
\end{figure}
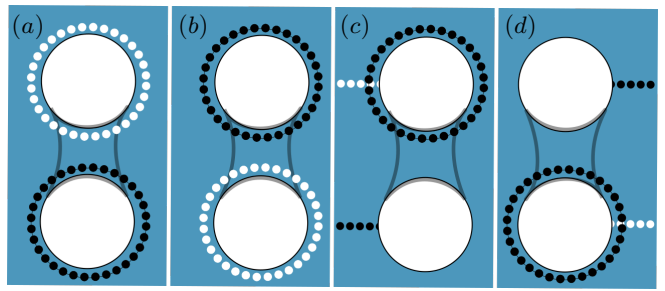

\emph{Clifford gates. ---}
We now turn our attention to performing Clifford gates on a qubit encoded in a wormhole.
We shall use an ancilla qubit initialized in a wormhole to perform single-qubit Clifford gates.

Suppose we have two qubits, labelled $1$ and $a$, denoting the qubit of interest and the ancilla respectively.
The following lemma summarizes what exactly is needed in order to perform single-qubit Clifford gates on qubit $1$.
\begin{restatable}{lem}{singlequbit}
\label{lem:claim}
	Let $\pa$ and $\pb$ be distinct, non-trivial single-qubit Pauli operators.
	Let $\pp$ and $\pq$ be two Pauli operators, not necessarily distinct.
	The two-qubit measurements $\pa_1\pp_a$ and $\pb_1\pq_a$, together with all single-qubit Pauli measurements on qubit $a$ are sufficient to generate the single-qubit Clifford group on qubit $1$.
\end{restatable}
The proof of this statement is presented in the appendix.

In addition to single-qubit Clifford gates, we need one entangling gate to generate the Clifford group.
This can be performed using just $\px$ and $\pz$ measurements, and an ancilla prepared in the $\ket{0}$ state as shown in fig. \ref{fig:cnot}.
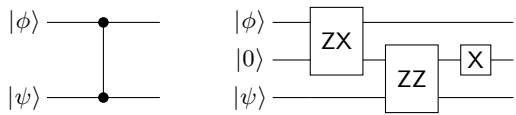
\begin{figure}[h]
	\centering
	\begin{tikzpicture}
		\node at (-3.8,0.5) {$\ket{\psi}$};
		\node at (-3.8,1.5) {$\ket{\phi}$};
		\draw (-3.5,0.5) edge[-] (-2,0.5);
		\draw (-3.5,1.5) edge[-] (-2,1.5);
		\fill[black] (-2.75,0.5) circle (0.07cm);
		\fill[black] (-2.75,1.5) circle (0.07cm);
		\draw (-2.75,0.5) edge[-] (-2.75,1.5);

		\node at (-0.8,0.5) {$\ket{\psi}$};
		\node at (-0.8,1) {$\ket{0}$};
		\node at (-0.8,1.5) {$\ket{\phi}$};
		\draw (-0.5,0.5) edge[-] (0,0.5);
		\draw (-0.5,1) edge[-] (0,1);
		\draw (-0.5,1.5) edge[-] (0,1.5);
		\draw (0,0.8) rectangle (0.7,1.7);
		\node at (0.35,1.25) {{$\pz\px$}};
		\draw (0.7,1) edge[-] (1,1);
		\draw (0,0.5) edge[-] (1,0.5);
		\draw (1,0.3) rectangle (1.7,1.2);
		\node at (1.35,0.75) {{$\pz\pz$}};
		\draw (1.7,1) edge[-] (2,1);
		\draw (2,0.8) rectangle (2.4,1.2);
		\node at (2.2,1) {$\px$};
		\draw (0.7,1.5) edge[-] (2.7,1.5);
		\draw (2.4,1) edge[-] (2.7,1);
		\draw (1.7,0.5) edge[-] (2.7,0.5);
	\end{tikzpicture}
	\caption{A circuit to perform $\cz$ between two single-qubit states $\psi$ and $\phi$ using measurements of Pauli operators.}
	\label{fig:cnot}
\end{figure}

We now need to demonstrate how to perform such a set of operations fault tolerantly.
To this end, we shall use an ancilla encoded in either a pair of smooth or rough punctures.
This ancilla, referred to as the needle, can be used to \emph{stitch} logical operators of interest as we shall demonstrate.
It will therefore not require any more long-range connectivity beyond what is required to initialize the wormholes.

There are different ways to entangle the needle and qubits encoded in the wormhole.
Braiding the needle around one mouth of a wormhole results in the controlled-$\pz$ operation between the needle and an encoded qubit.
The evolution of the logical $\px$ operator of the puncture is shown in fig. \ref{fig:braid-needle}.
\begin{figure}[h]
  \centering
  \begin{tikzpicture}
	\node at (0,0) {\includegraphics[width=\columnwidth]{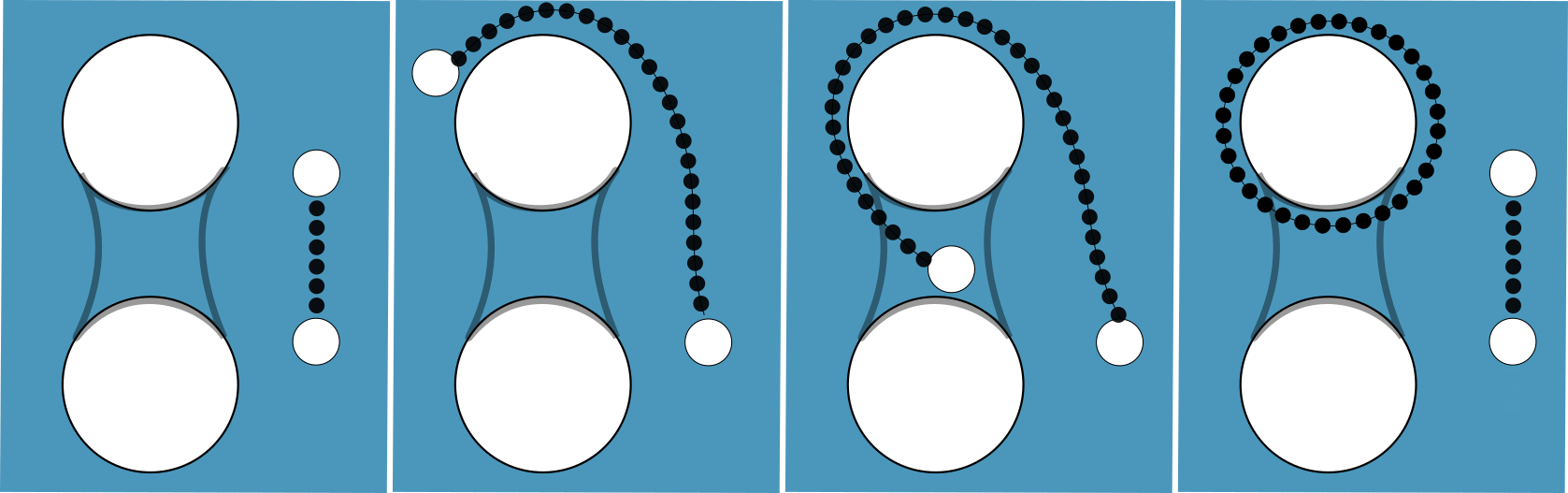}};
	\node at (-4.1,-1.2) {$(a)$};
	\node at (-1.9,-1.2) {$(b)$};
	\node at (0.25,-1.2) {$(c)$};
	\node at (2.4,-1.2) {$(d)$};
\end{tikzpicture}
  \caption{Braiding the needle around the wormhole results in a controlled-$\pz$ operation between the needle and encoded qubit.}
  \label{fig:braid-needle}
\end{figure}
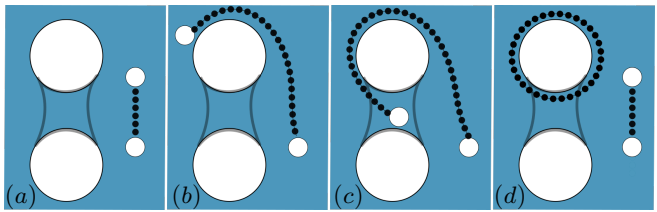

Since the wormhole is traversable, a puncture can enter one mouth of the wormhole and emerge via the other.
We call this operation \emph{stitching}.
Stitching results in the controlled-$\px$ operation between the needle and the encoded qubit.
The evolution of the logical $\px$ of the puncture is shown in fig. \ref{fig:stitch-needle}.
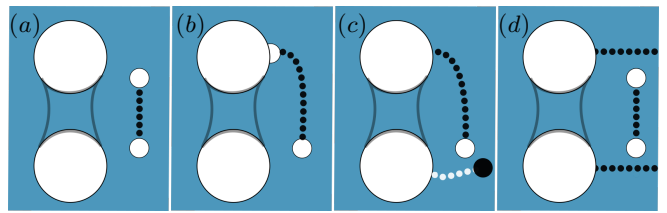
\begin{figure}[h]
	\centering
	\begin{tikzpicture}
		\node at (0,0) {\includegraphics[width=\columnwidth]{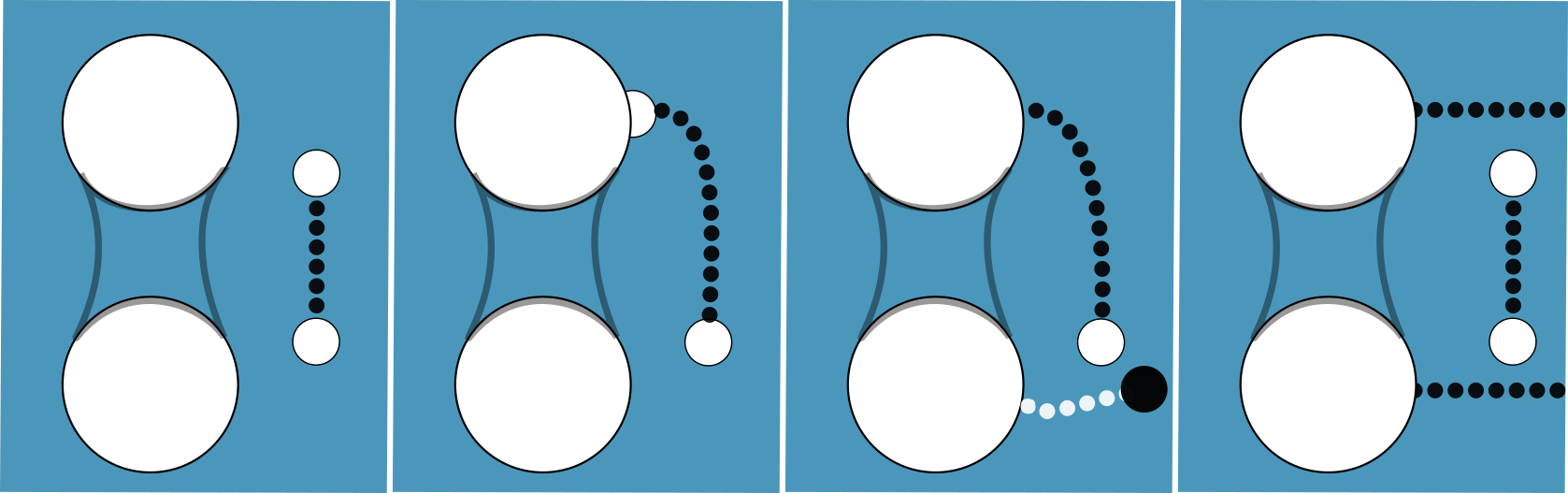}};
		\node at (-4.1,1.1) {$(a)$};
		\node at (-1.9,1.1) {$(b)$};
		\node at (0.25,1.1) {$(c)$};
		\node at (2.4,1.1) {$(d)$};
	\end{tikzpicture}
	\caption{Passing the puncture through the mouth of the wormhole results in the controlled-$\px$ between the needle and encoded qubit.
	Between panels $(c)$ and $(d)$, the puncture goes through the sink wormhole so as to return to the appropriate type.}
	\label{fig:stitch-needle}
\end{figure}

We define certain properties of the logical operators.
First we notice that the needle operators $\px$ and $\pz$ are efficiently preparable and can be measured fault tolerantly.
If we need to measure the string-like operator that runs between punctures for instance, we could make the punctures larger and bring them closer together.
Alternatively to measure the loop-type operator, we can move the punctures apart, make them small and measure the boundary.
The logical $\px$ and $\pz$ operators will thus be referred to as \emph{needle-measurable} operators.
On the other hand, the logical $\py$ operator associated to a puncture is not needle-measurable as this would necessitate shrinking the boundary as well as bringing the punctures close together.
In turn, the measurement would no longer be fault tolerant.

An operator $\pq$ is \emph{traceable} if there exists a way to map a needle-measurable operator $\pp$ to $\pp\pq$.
The logical $\pz$ operator of a puncture was already traceable.
We highlight that by converting a puncture to a wormhole, the logical $\px$ operator is now also traceable.

The final ingredient required to perform logical single-qubit Clifford gates as stipulated by lemma \ref{lem:claim} is a $\py$ measurement.
Unfortunately, logical $\py$ operators of wormholes are not traceable as the $\py$ operator crosses itself.
To be precise, let $\py_w$ denote the logical $\py$ of the wormhole and $\py_n$ be the logical $\py$ of the needle.
In following the path of a logical $\py_w$ associated to a wormhole, we find that it is the logical $\py_n$ operator of the needle that is mapped to $\py_n\py_w$.
Since the logical $\py_n$ of the needle is not needle-measurable, the logical $\py$ of the wormhole is not traceable.

This is remedied with a resource state as follows.
We let a wormhole that can encode $2$ qubits serve as the ancilla.

\begin{restatable}{lem}{catalytic}
\label{lem:catalytic}
	Let the ancilla be comprised of two qubits labelled $a$ and $b$ such that one of its stabilizer generators is $\i_a\py_b$.
	It is possible to apply the measurement $\py_{a}\i_{b}$ on qubit $a$ without affecting the state of the generator $\i_a\py_b$.
\end{restatable}

The intuition behind this claim is that the operator $\py_{a}\py_{b}$ does not cross itself and hence is traceable.
Assuming such a resource state is provided ahead of time, we can perform catalytic Clifford gates.
This satisfies the requirements for performing single-qubit Clifford gates.
Stitching and braiding have already shown that it is possible to entangle the qubits encoded in a wormhole with those in a puncture.
We can extend this to entangle two qubits encoded in wormholes and create an entangling gate.
This therefore generates all gates in the Clifford group.

\emph{Discussion and conclusion. ---}
We have introduced a new defect on the toric code called a wormhole using entangled measurements along the boundaries of punctures.
These defects are a unified representation of both puncture and twist type defects.
The stabilizers within the interior of the mouths of the wormholes have been removed from the code much like in the case of punctures.
By entangling the boundaries of these punctures, we see interesting physics when we consider the movement of anyons on the surface of the lattice.
Wormholes are capable of encoding a logical qubit and we can perform all gates in the Clifford group using topologically non-trivial operations.
Importantly, wormholes provide a way to perform fault-tolerant gates on a class of quantum LDPC codes called hypergraph product codes.

\emph{Acknowledgements. ---}
We would like to thank Colin Trout for detailed comments on an earlier draft of this work, Baptiste Royer for discussions on deterministic entanglement sharing in superconducting qubit architectures and Ben Criger for discussions on code deformation.
A.K. acknowledges support from the Fonds de recherche du Qu\'ebec - Nature et technologies (FRQNT) via the B2X scholarship for doctoral candidates. David Poulin is a CIFAR Fellow with the Quantum Information Science program.

\bibliographystyle{unsrtabbrev}
\bibliography{references}

\begin{thebibliography}{10}

\bibitem{kitaev2003fault}
A.~Y. Kitaev.
\newblock Fault-tolerant quantum computation by anyons.
\newblock {\em Annals of Physics}, 303(1):2--30, 2003.

\bibitem{bravyi1998quantum}
S.~B. Bravyi and A.~Y. Kitaev.
\newblock Quantum codes on a lattice with boundary.
\newblock {\em arXiv preprint quant-ph/9811052}, 1998.

\bibitem{bombin2006topological}
H.~Bombin and M.~A. Martin-Delgado.
\newblock Topological quantum distillation.
\newblock {\em Physical Review Letters}, 97(18):180501, 2006.

\bibitem{FowlerEtAl12}
A.~G. Fowler, M.~Mariantoni, J.~M. Martinis, and A.~N. Cleland.
\newblock Surface codes: Towards practical large-scale quantum computation.
\newblock {\em Physical Review A}, 86(3):032324, 2012.

\bibitem{bravyi2010tradeoffs}
S.~Bravyi, D.~Poulin, and B.~Terhal.
\newblock Tradeoffs for reliable quantum information storage in 2{D} systems.
\newblock {\em Physical Review Letters}, 104(5):050503, 2010.

\bibitem{bravyi2013classification}
S.~Bravyi and R.~K{\"o}nig.
\newblock Classification of topologically protected gates for local stabilizer
  codes.
\newblock {\em Physical Review Letters}, 110(17):170503, 2013.

\bibitem{webster2018braiding}
P.~Webster and S.~D. Bartlett.
\newblock Braiding defects in topological stabiliser codes of any dimension
  cannot be universal.
\newblock {\em arXiv preprint arXiv:1811.11789}, 2018.

\bibitem{nickerson2013topological}
N.~H. Nickerson, Y.~Li, and S.~C. Benjamin.
\newblock Topological quantum computing with a very noisy network and local
  error rates approaching one percent.
\newblock {\em Nature communications}, 4:1756, 2013.

\bibitem{kurpiers2018deterministic}
P.~Kurpiers, P.~Magnard, T.~Walter, B.~Royer, M.~Pechal, J.~Heinsoo,
  Y.~Salath{\'e}, A.~Akin, S.~Storz, J.-C. Besse, et~al.
\newblock Deterministic quantum state transfer and remote entanglement using
  microwave photons.
\newblock {\em Nature}, 558(7709):264, 2018.

\bibitem{axline2018demand}
C.~J. Axline, L.~D. Burkhart, W.~Pfaff, M.~Zhang, K.~Chou, P.~Campagne-Ibarcq,
  P.~Reinhold, L.~Frunzio, S.~Girvin, L.~Jiang, et~al.
\newblock On-demand quantum state transfer and entanglement between remote
  microwave cavity memories.
\newblock {\em Nature Physics}, page~1, 2018.

\bibitem{campagne2018deterministic}
P.~Campagne-Ibarcq, E.~Zalys-Geller, A.~Narla, S.~Shankar, P.~Reinhold,
  L.~Burkhart, C.~Axline, W.~Pfaff, L.~Frunzio, R.~Schoelkopf, et~al.
\newblock Deterministic remote entanglement of superconducting circuits through
  microwave two-photon transitions.
\newblock {\em Physical Review Letters}, 120(20):200501, 2018.

\bibitem{katzgraber2014glassy}
H.~G. Katzgraber, F.~Hamze, and R.~S. Andrist.
\newblock Glassy chimeras could be blind to quantum speedup: Designing better
  benchmarks for quantum annealing machines.
\newblock {\em Physical Review X}, 4:021008, Apr 2014.

\bibitem{raussendorf2007fault}
R.~Raussendorf and J.~Harrington.
\newblock Fault-tolerant quantum computation with high threshold in two
  dimensions.
\newblock {\em Physical Review Letters}, 98(19):190504, 2007.

\bibitem{Bombin10}
H.~Bombin.
\newblock Topological order with a twist: Ising anyons from an abelian model.
\newblock {\em Physical Review Letters}, 105(3):030403, 2010.

\bibitem{Bombin11}
H.~Bombin.
\newblock Clifford gates by code deformation.
\newblock {\em New Journal of Physics}, 13(4):043005, 2011.

\bibitem{BrownEtAl17}
B.~J. Brown, K.~Laubscher, M.~S. Kesselring, and J.~R. Wootton.
\newblock Poking holes and cutting corners to achieve clifford gates with the
  surface code.
\newblock {\em Physical Review X}, 7(2):021029, 2017.

\bibitem{FreedmanMeyerLuo02}
M.~H. Freedman, D.~A. Meyer, and F.~Luo.
\newblock Z2-systolic freedom and quantum codes.
\newblock {\em Mathematics of quantum computation, Chapman \& Hall/CRC}, pages
  287--320, 2002.

\bibitem{GuthLubotzky14}
L.~Guth and A.~Lubotzky.
\newblock Quantum error correcting codes and 4-dimensional arithmetic
  hyperbolic manifolds.
\newblock {\em Journal of Mathematical Physics}, 55(8):082202, 2014.

\bibitem{TillichZemor09}
J.-P. Tillich and G.~Z{\'e}mor.
\newblock Quantum {LDPC} codes with positive rate and minimum distance
  proportional to the square root of the blocklength.
\newblock {\em IEEE Transactions on Information Theory}, 60(2):1193--1202,
  2014.

\bibitem{krishna2019fault}
A.~Krishna and D.~Poulin.
\newblock Fault-tolerant gates on hypergraph product codes.
\newblock 2019.

\bibitem{LeverrierTillichZemor15}
A.~Leverrier, J.-P. Tillich, and G.~Z{\'e}mor.
\newblock Quantum expander codes.
\newblock In {\em Foundations of Computer Science (FOCS), 2015 IEEE 56th Annual
  Symposium on}, pages 810--824. IEEE, 2015.

\bibitem{pastawski2015holographic}
F.~Pastawski, B.~Yoshida, D.~Harlow, and J.~Preskill.
\newblock Holographic quantum error-correcting codes: Toy models for the
  bulk/boundary correspondence.
\newblock {\em Journal of High Energy Physics}, 2015(6):149, 2015.

\bibitem{susskind2016copenhagen}
L.~Susskind.
\newblock Copenhagen vs {E}verett, teleportation, and {ER} $=$ {EPR}.
\newblock {\em Fortschritte der Physik}, 64(6-7):551--564, 2016.

\bibitem{bekenstein1973black}
J.~D. Bekenstein.
\newblock Black holes and entropy.
\newblock {\em Physical Review D}, 7(8):2333, 1973.

\bibitem{yoder2017surface}
T.~J. Yoder and I.~H. Kim.
\newblock The surface code with a twist.
\newblock {\em Quantum}, 1:2, 2017.

\bibitem{zheng2015demonstrating}
H.~Zheng, A.~Dua, and L.~Jiang.
\newblock Demonstrating non-abelian statistics of majorana fermions using twist
  defects.
\newblock {\em Physical Review B}, 92(24):245139, 2015.

\bibitem{barkeshli2016modular}
M.~Barkeshli and M.~Freedman.
\newblock Modular transformations through sequences of topological charge
  projections.
\newblock {\em Physical Review B}, 94(16):165108, 2016.

\end{thebibliography}

\section{Appendix}
\label{sec:appendix}
\subsection*{Proof of lemma \ref{lem:claim}}
We include the proof of lemma \ref{lem:claim} here, which is restated here for completeness.
\singlequbit*
\begin{proof}
	A logical Clifford operation proceeds in three steps.
	Without loss of generality, let $\pp = \pq = \pa$ and consider the measurement of $\pa_1\pp_a$.
	\begin{enumerate}
		\item Initialize qubit $a$ by preparing it in the $\pb$ basis.
		\item Next, perform a joint measurement $\pa_1\pp_a (=\pa_1\pa_a)$ of qubits $1$ and $a$.
		\item Finally, measure qubit $a$ in the basis $\pc (\neq \pa \neq \pb \neq \i)$.
	\end{enumerate}
	The following flowchart tracks the transformation of the generators of the associated stabilizer and normalizer groups, $\S$ and $\N$.
	\begin{align*}
		\begin{matrix}
		\S &=& \{\pb_a\} &\to& \{\pa_1 \pa_a \} &\to& \{\pc_a\} \\
		\N &=& \{\pb_1,\pc_1\} &\to& \{\pb_1\pb_a, \pc_1\pb_a\} &\to& \{\pc_1\pc_a, \pb_1\pc_a\}~.
		\end{matrix}
	\end{align*}
	We have used the fact that Pauli operators are cyclic, i.e. the product of any two distinct operators yields the third (up to a phase).
	Up to stabilizer, the result of this transformation is to map $\pb$ to $\pc$ and vice-versa.
	The result follows.
\end{proof}

\subsection*{Catalytic Clifford gates}
\catalytic*
\begin{proof}
Let $a$ and $b$ refer to the qubits encoded in a wormhole.
However, the product $\py_{a}\py_{b}$ is traceable as shown in fig. \ref{fig:xxyy}.
\begin{figure}[htp]
	\centering
	\begin{tikzpicture}
		\node at (0,0) {\includegraphics[scale=0.3]{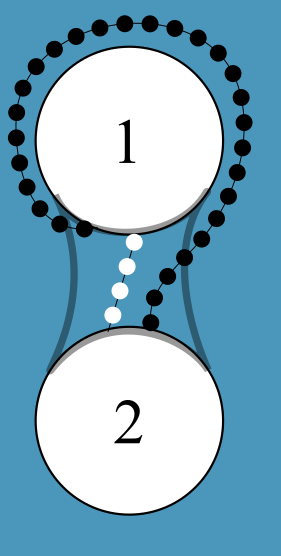}};
	\end{tikzpicture}
	\caption{The product $\py_{a}\py_{b}$ does not cross itself and is therefore traceable.}
	\label{fig:xxyy}
\end{figure}
This is because the operator does not intersect itself.
This can be used to measure $\py_{a}\i_{b}$ by initializing the wormhole in a state such that $\i_{a}\py_{b}$ is a stabilizer generator.
We can then measure $\py_{a}\py_{b}$, which up to action of an element of the stabilizer, is equivalent to $\py_{a}\i_{b}$.
The generator $\i_{a}\py_{b}$ commutes with the measurement and is therefore unaffected.
It can therefore be used for the next gate as well and in this sense, the gate is catalytic.
\end{proof}

\end{document}